\documentclass[utf8,latin1]{article}
\usepackage[utf8]{inputenc}
\usepackage{amsfonts,amsmath,amsthm,amssymb,bm,wasysym,comment}
\usepackage{enumerate}

\newtheorem{theorem}{Theorem}
\newtheorem{remark}[theorem]{Remark}
\newtheorem{lemma}[theorem]{Lemma}
\newtheorem{corollary}[theorem]{Corollary}

\newtheorem{example}[theorem]{Example}
\newtheorem{notation}[theorem]{Notation}

\newtheorem*{mrpc}{MinRank Problem}

\newtheorem*{mdd}{Minimum Rank-Distance Decoding}

\usepackage[autostyle=true, italian=guillemets] {csquotes}
\usepackage[english]{babel}

\usepackage{graphics,standalone,tikz,pgfplots,textcomp,tikz}
\usepackage{geometry}
\DeclareMathOperator{\Syz}{Syz}
\DeclareMathOperator{\grade}{grade}

\DeclareMathOperator{\rk}{rk}

\newcommand{\FF}{\mathbb{F}}
\newcommand{\KK}{\mathbb{K}}
\newcommand{\ZZ}{\mathbb{Z}}

\newcommand{\abs}[1]{\ensuremath{\left\lvert #1 \right\rvert}}

\renewcommand{\vec}[1]{\ensuremath{ \mathbf #1 }}

\usepackage{hyperref} 
\hypersetup{pdfpagemode=UseNone,colorlinks=true,linkcolor=blue,citecolor=blue,urlcolor=blue,allbordercolors=white,pdftitle={Reg LPP-ideals},pdfauthor={}}

\definecolor{bluegray}{rgb}{0.4, 0.6, 0.8}

\title{The complexity of the SupportMinors Modeling for the MinRank Problem}

\author{Daniel Cabarcas\footnote{The first author was funded by a CIMPA-ICTP Fellowships ``Research in Pairs".} \and Giulia Gaggero \and Elisa Gorla\thanks{The second and third author were funded by by armasuisse through grant no. CYD-C-2020010.}}
\date{}

\begin{document}

\maketitle

\begin{abstract} 
In this note, we provide proven estimates for the complexity of the SupportMinors Modeling, mostly confirming the heuristic complexity estimates contained in the original article.
\end{abstract}

\section{Introduction}

The MinRank Problem arises naturally within cryptography and coding theory, as well as in numerous other applications. The problem in its general form can be stated as follows.

\begin{mrpc}
Let $\KK$ be a field and let $m,n,k$ be positive integers. Given as input $K$ matrices $M_1,\dots,M_K\in\KK^{n\times m}$, find $x_1,\dots,x_K\in\KK$ such that the matrix $\sum_{\ell=1}^kx_\ell M_\ell$ is nonzero and has least possible rank. 
\end{mrpc}

In situations when the least possible rank (or a tight upper bound for it) is known, one may rephrase the problem as follows.

\begin{mrpc}
Let $\KK$ be a field and let $m,n,r,k$ be positive integers. Given as input $K$ matrices $M_1,\dots,M_K\in\KK^{n\times m}$, find $x_1,\dots,x_K\in\KK$ such that 
\[
0<\rk\left(\sum_{\ell=1}^K x_\ell M_\ell\right)\leq r.
\] 	 
\end{mrpc}

The MinRank Problem plays a central role within multivariate cryptography, both in the cryptanalysis and in the design of schemes, as it is NP-complete~\cite{BFS96} and believed to be quantum-resistant. For example, it is a central tool in the cryptanalysis of HFE and its variants~\cite{KS99, BFP13, CSV17, VS17, DPPS18}, the TTM Cryptosystem~\cite{GC00}, and the ABC Cryptosystem~\cite{MSP14, MPS17}. More recently, both GeMSS and Rainbow were subject to attacks which exploit the MinRank Problem, see~\cite{B21a,TPD21,BBCPSV21,B22}. 

On the constructive side, a zero-knowledge protocol based on the MinRank Problem was proposed by Courtois in~\cite{C01}. This produces a signature scheme following~\cite{FS87}. Recently, a scheme relying on the MinRank Problem for its security was proposed in~\cite{RLT21} and cracked in~\cite{BTV21}. Digital signature relying on the MinRank Problem for their security were submitted to the NIST Call for Additional Digital Signature Schemes in 2023~\cite{ARVBBESZ23-MiRitH,ABCFGNR23-MIRA}, see also~\cite{BESJ22}. 

In addition, the MinRank Problem is closely related to Minimum Distance Decoding in the rank-metric. Notice in fact that, if $\mathcal{C}\subseteq\KK^{n\times m}$ is a linear rank-metric code with basis $M_1,\dots,M_K\in\KK^{n\times m}$ and minimum distance $d(\mathcal{C})$, then Minimum Distance Decoding in the rank-metric can be phrased as follows.

\begin{mdd}
Given a received matrix $M_0\in\KK^{n\times m}$ and a basis $M_1,\dots,M_K\in\KK^{n\times m}$ of the code $\mathcal{C}$, find $x_1,\dots,x_K\in\KK$ such that
\[
\rk\left(\sum_{\ell=1}^K x_\ell M_\ell-M_0\right)\leq\frac{d(\mathcal{C})-1}{2}.
\] 	
\end{mdd}

One sees immediately that Minimum Rank-Distance Decoding is an instance of the MinRank Problem in its second formulation. 
Therefore, estimates on the complexity of the MinRank Problem have a direct impact on understanding the complexity of decoding a general code with respect to the rank-metric and hence on complexity estimates in rank-metric code-based cryptography. Similarly to the MinRank Problem, Minimum Rank-Distance Decoding is known to be NP-hard~\cite{C01} and believed to be quantum-resistant.

Rank-metric code-based cryptography may be traced back to~\cite{GPT91}, where Gabidulin codes were used. The weakness of this proposal and of some related attempts were later exposed in~\cite{O08, FL05, BL04, L10}. Rank-metric code-based cryptography has become relevant again in recent years, when several new cryptographic schemes based on Minimum Rank-Distance Decoding were proposed in the context of the NIST Post-Quantum Cryptography Standardization process, see~\cite{ABDGHRTZ17-LAKE, ABDGHRTZ17-LOCKER, MABBBDGHHZB17-OUROBOROS, MABBBBD19-ROLLO, MABBBDGZ17-RQC, GRSZ14-RANKSIGN, ABGHZ19-DURANDAL}. While none of these proposals was selected, NIST expressed interest in further study of rank-based cryptosystems. In addition, two digital signatures~\cite{CNPRRST23-MEDS, ABBCDFGJRT23-RYDE} which base their security on Minimum Rank-Distance Decoding were submitted this year to the NIST Call for Additional Digital Signature Schemes.

Main modelings for the MinRank Problem are the Kipnis-Shamir~\cite{KS99} and the Minors Modeling~\cite{FLP08, FFS10, FFS13, CG21}. In the past few years, variations of the Minors Modeling such as the MaxMinors Modeling~\cite{BBBGNRT20} and the SupportMinors Modeling~\cite{BBCGPSTV20,BB22} were introduced. Their advantage is that experimentally they appear to have significantly lower complexity compared to the classical Minors Modeling. However, their complexity is significantly less well-understood and current estimates heavily rely on heuristic assumptions.

This work aims at making the heuristics of~\cite{BBCGPSTV20} rigorous, therefore establishing a rigorous upper bound on the complexity of MinRank. In Section~\ref{sect:prelim}, we fix the notation and review some preliminaries that will be used throughout the paper. In Section~\ref{sect:generic} we study the algebraic situation in which the entries of the matrix are distinct variables. The results obtained in Section~\ref{sect:generic} are used in Section~\ref{sect:spec}, where we treat the general case. Our main result is Theorem~\ref{thm:main}, where we compute the dimension of the vector space of relations among the rows of the Macaulay matrix used in the SupportMinors Algorithm. Using the results from the previous sections, in Section~\ref{sect:estimates} we provide rigorous estimates for the complexity of the SupportMinors Modeling. Our findings mostly confirm the heuristic estimates from~\cite{BBCGPSTV20}.

\section{Preliminaries and notation}\label{sect:prelim}

Let $\mathbb{K}$ be a field and let $r,m,n,K$ be positive integers, $r\leq n$. 
For a positive integer $t$, denote by $[t]$ the set $\{1,\ldots,t\}$.
For a square matrix $M$, denote by $|M|$ the determinant of $M$. 
Throughout the paper, we work in the polynomial rings $$R=\KK[y_{kj},c_{ij}: k\in[m], i\in[r], j\in[n]] \mbox{ and } P=\KK[x_{\ell},c_{ij}: \ell\in[K], i\in[r], j\in[n]],$$ where $x_1,\ldots,x_K$, $y_{11},\ldots,y_{mn}$, and $c_{11},\ldots,c_{rn}$ are distinct variables.

\paragraph{SupportMinors}
Consider an instance of the MinRank Problem $M_1,\ldots,M_K$ with target rank $r\geq 1$ and
let $M_\vec{x}=\sum_{\ell=1}^K x_\ell M_\ell$. Let $C=(c_{ij})$ be an $r\times n$ matrix whose entries are distinct variables. SupportMinors addresses the MinRank Problem by solving the polynomial system in $x_\ell$ and $c_{ij}$ consisting of the equations

\begin{equation}\label{eqn:sm}
\left\{
\left|\begin{pmatrix}
	r_i \\
	C
\end{pmatrix}\right|_{*,J}:
J\subseteq [n], \abs{J} = r+1, r_i \text{ $i$-th row of } M_\vec{x}, i\in[r]
\right\}.
\end{equation}
Here we denote by $|M|_{*,J}$ the maximal minor of $M$ corresponding to the columns indexed by $J$. 
Notice that each entry of $M_\vec{x}$ is a linear form $m_{kj}(\vec{x}) = \sum_{\ell=1}^K m^\ell_{k j}x_\ell$, where $m^\ell_{k j}$ is the entry of $M_\ell$ in position $(k,j)$.

\bigskip

In order to study system (\ref{eqn:sm}), in Section \ref{sect:generic} we analyze the simpler situation when the matrix $M_\vec{x}=(m_{kj}(\vec{x}))$ is replaced by the matrix $Y=(y_{kj})$, whose entries are distinct variables. Let
$$D = 
\begin{pmatrix}
	Y \\
	C
\end{pmatrix}.$$
Let $\mathcal{F} \subseteq R$ be the set of the $(r+1)$-minors of $D$ and let $\mathcal{G}$ be the set of $(r+1)$-minors of $D$ that involve the last $r$ rows. 

In Section \ref{sect:spec}, we apply the results of Section \ref{sect:generic} to the study of the SupportMinors system (\ref{eqn:sm}).
For that, let $\rho:R\rightarrow P$ be the $R$-algebra homomorphism defined by $c_{ij}\mapsto c_{ij}$ and $y_{kj}\mapsto m_{kj}(\vec{x})$. Abusing notation, we denote by $\rho$ also the homomorphism $R^t \rightarrow P^t$, $t \in \mathbb{N}$, that acts as $\rho$ componentwise. We often use specializing as a synonym for computing the image of an object via $\rho$. Then $$\rho(D)=\begin{pmatrix}
	M_{\vec{x}} \\
	C
\end{pmatrix}$$ 
and $\rho(\mathcal{G})$ coincides with (\ref{eqn:sm}).

Let $I\subseteq [m+r]$ and $J\subseteq [n]$ be multisets with $|I|=|J|$. Throughout the paper, we abuse notation and write that a multiset is included in $[u]$ if its underlying set is included in $[u]$. 
For a matrix $D$, denote by $D_{I,J}$ the submatrix consisting of the rows indexed by the elements of $I$ and of the columns indexed by the elements of $J$, where a row or column appears with the same multiplicity as it appears in the corresponding index multiset.
Finally, for $I\subseteq [m+r]$ and $J\subseteq [n]$ subsets with $|I|=|J|=r+1$, let $E_{I,J}$ be the corresponding standard basis vector of $R^{\binom{r+m}{r+1} \binom{n}{r+1}}$. We often refer to the position of the only nonzero entry of $E_{I,J}$ as position $(I,J)$.

Define the map
$$\begin{array}{rcl}
\phi:R^{\binom{r+m}{r+1} \binom{n}{r+1}} & \rightarrow & R \\
E_{I,J} & \mapsto & |D_{I,J}|.
\end{array}$$
The image of the standard basis of $R^{\binom{r+m}{r+1} \binom{n}{r+1}}$ is $\mathcal{F}$ and the syzygy module of $\mathcal{F}$ is $$\Syz(\mathcal{F})=\ker(\phi).$$
Similarly, the syzygy module of $\mathcal{G}$ is $$\Syz(\mathcal{G})=\ker(\varphi),$$ where $\varphi$ is the restriction of $\phi$ to the free submodule of $R^{\binom{r+m}{r+1} \binom{n}{r+1}}$ generated by $\{E_{I,J}: I\supseteq[m+1,\ldots,m+r]\}.$
This is precisely the subset of the standard basis of $R^{\binom{r+m}{r+1} \binom{n}{r+1}}$ whose elements are mapped to the elements of $\mathcal{G}$ via $\phi$. 
For more information on the syzygy module and related concepts, we refer the interested reader to \cite[Chapter 2.3]{KR00}. We have $$\Syz(\mathcal{G})\subseteq\Syz(\mathcal{F})\subseteq R^{\binom{r+m}{r+1} \binom{n}{r+1}},$$ where $\Syz(\mathcal{G})$ denotes the syzygies of $\mathcal{F}$, which only involve the elements of $\mathcal{G}$. 
We are interested in the submodule
$$\mathbf{U}=\Syz(\mathcal{G}) \cap \KK[y_{kl}: 1 \leq k \leq m, \mbox{ } 1 \leq l \leq n]^{m \binom{n}{r+1}}$$
of syzygies of $\mathcal{G}$ which only involve the $y$-variables.
Our motivation for looking at this specific submodule comes from the algorithm proposed in~\cite{BBCGPSTV20} and will become clear in Section~\ref{sect:estimates}. 
In order to estimate the complexity of the SupportMinors Algorithm, we need to compute the dimension of certain graded components of $$\Syz(\rho(\mathcal{G}))\cap \KK[x_1,\ldots,x_{\ell}]^{m \binom{n}{r+1}}\subseteq \Syz(\rho(\mathcal{G})).$$
Clearly $\rho(\mathbf{U})\subseteq \Syz(\rho(\mathcal{G}))\cap \KK[x_1,\ldots,x_{\ell}]^{m \binom{n}{r+1}}$, however in Theorem \ref{thm:main} we will show that equality holds in degree $r+2$.

\section{Matrices of variables}\label{sect:generic}

We start by studying the simpler situation in which we replace the matrix $M_x$ by a matrix $Y$ whose entries are distinct variables.
For this, let $I\subseteq [m+r]$ and $J\subseteq [n]$ be ordered multisets with $|I|=|J|=r+2$. Let $1\leq i,j\leq |I|$. If the $i$-th element of $I$ appears more than once in $I$, then developing the determinant of $D_{I,J}$ with respect to the $i$-th row yields a syzygy of $\mathcal{F}$, that we denote by $|D_{I,J}|_{i,\bullet}$. 
Similarly, if the $j$-th element of $J$ appears more than once in $J$, then developing the determinant of $D_{I,J}$ with respect to the $j$-th column yields a syzygy of $\mathcal{F}$, that we denote by $|D_{I,J}|_{\bullet,j}$.
Finally, if $I$ and $J$ are sets, then the difference of an expansion of $|D_{I,J}|$ with respect to the $i$-th and the $j$-th row produces a syzygy of $\mathcal{F}$, which we denote by $|D_{I,J}|_{i,\bullet}-|D_{I,J}|_{j,\bullet}.$ The difference of an expansion of $|D_{I,J}|$ with respect to the $i$-th row and the $j$-th column also produces a syzygy of $\mathcal{F}$, which we denote by $|D_{I,J}|_{i,\bullet}-|D_{I,J}|_{\bullet,j}.$ 

\begin{remark}
When writing $D_{I,J}$, we think of $I$ and $J$ as ordered multiset. However, the order only affects the determinant by a sign, hence any reordering of $I$ and $J$ produces an equivalent syzygy. Throughout the paper we ignore the ordering, in the hope that this does not confuse the reader.
\end{remark}

In~\cite[Theorem 5.1]{K98}, Kurano proved that the following elements are a system of generators of $\Syz(\mathcal{F})$ as an $R$-module.

\begin{itemize}
\item[Type I:] For each $I\subseteq [m+r]$ and $J\subseteq [n]$ subsets with $|I|=r+1$, $|J|=r+2$ and for each $h\in I$, one has a syzygy $|D_{\{h\}\cup I,J}|_{1,\bullet}$, where $D_{\{h\}\cup I,J}$ is the matrix obtained from $D_{I,J}$ by adding a copy of the $h$-th row as first row.
Similarly, exchanging the roles of rows and columns, for each $I\subseteq [m+r]$ and $J\subseteq [n]$ subsets with $|I|=r+2$, $|J|=r+1$ and for each $k\in J$, one has a syzygy $|D_{I,\{k\}\cup J}|_{\bullet,1}$, where $D_{I,\{k\}\cup J}$ is the matrix obtained from $D_{I,J}$ by adding a copy of the $k$-th column as first column.
\item[Type II:] For each $I\subseteq [m+r]$ and $J\subseteq [n]$ subsets with $|I|=|J|=r+2$ and for each $h,k\in[r+2]$, one has a syzygy $|D_{I,J}|_{h,\bullet}-|D_{I,J}|_{\bullet,k}$.
\end{itemize}

Notice that all the above relations yield linear syzygies. Moreover, since the syzygies are homogeneous of degree $r+2$, then $\Syz(\mathcal{F})_{r+2}$ is generated as a $\KK$-vector space by the same syzygies.

\medskip
	
The polynomial ring $R$ can be given a standard $\mathbb{Z}^{m+r} \oplus \mathbb{Z}^n$-grading \textquoteleft\textquoteleft by rows and columns\textquoteright\textquoteright~by setting $\deg(y_{kj})=e_k+f_j \in \mathbb{Z}^{m+r} \oplus \mathbb{Z}^n$ and $\deg(c_{ij})=e_{m+i}+f_j \in \mathbb{Z}^{m+r} \oplus \mathbb{Z}^n$, where $\{e_1,\ldots,e_{m+r}\}$ is the standard basis of $\ZZ^{m+r}$ and $\{f_1,\ldots,f_n\}$ that of $\ZZ^{n}$. The multidegree of a monomial $\mu=\prod_{i=1}^{r}\prod_{k=1}^m\prod_{l=1}^n\prod_{j=1}^n y_{kl}^{\alpha_{kl}} c_{ij}^{\beta_{ij}} \in R$, where $\alpha_{kl}, \beta_{ij} \in \mathbb{Z}_{\geq0}$, is 

$$\deg(\mu)=\sum_{k=1}^{m}\sum_{l=1}^{n} \alpha_{kl} e_{k} + \sum_{i=1}^{r}\sum_{j=1}^{n} \beta_{ij} e_{m+i} + \sum_{l=1}^{n}\sum_{k=1}^{m} \alpha_{kl} f_{l} + \sum_{j=1}^{n}\sum_{i=1}^{r} \beta_{ij} f_{l}\in \ZZ^{m+r} \oplus \ZZ^{n}.$$
We often use the word multigraded to mean homogenoeus with respect to the multigrading.
Notice that the polynomials in $\mathcal{F}$ are multigraded. In fact, the minor that involves the rows and columns indexed by $I$ and $J$ has multidegree
$$\deg(|D_{I,J}|)=\sum_{i \in I} e_i + \sum_{j \in J} f_{j}.$$
Every minor in $\mathcal{G}$ involves the last $r$ rows of $D$, hence it corresponds to an $I$ of the form $I=\{h,m+1,\ldots,m+r\}$ for some $h\in[m]$. Therefore it is homogeneous of multidegree $$\deg(|D_{I,J}|)=e_h + \sum_{i=m+1}^{m+r} e_i + \sum_{j \in J} f_{j},$$ for some $h \in [m]$ and subset $J \subseteq [ n ]$ with $\abs{J} = r+1$.

\medskip
 
We can divide Kurano's syzygies in four disjoint sets:
\begin{align*}
\mathbf{S}_1=& \left\{
|D_{\{h\}\cup I,J}|_{1,\bullet}: \abs{I}=r+1, \abs{J} = r+2, h \in I \right\},\\	
\mathbf{S}_2=& \left\{
|D_{I,\{k\}\cup J}|_{\bullet,1}
: \abs{I} =r+2, \abs{J} = r+1, k \in J \right\},\\
\mathbf{S}_3=& \left\{
|D_{I,J}|_{1,\bullet}-|D_{I,J}|_{h,\bullet}
: \abs{I} = \abs{J}=r+2, 2\leq h \leq r+2\right\},\\
\mathbf{S}_4=& \left\{
|D_{I,J}|_{1,\bullet} - |D_{I,J}|_{\bullet,k}: \abs{I} = \abs{J} =r+2, 2\leq k\leq r+2\right\}
\end{align*}
Notice that $\mathbf{S}_3 \cup \mathbf{S}_4$ is smaller than the set of Type II syzygies, however it is an easy exercise to check that every Type II syzygy is a linear combination of elements of $\mathbf{S}_3 \cup \mathbf{S}_4$. The $\mathbf{S}_4$-type syzygies corresponding to $k=1$, in particular, is the sum with alternating signs of the elements of $\mathbf{S}_3$ minus the sum with alternating signs of the elements of $\mathbf{S}_4$. Therefore, the set $$\mathbf{S}=\mathbf{S}_1 \cup \mathbf{S}_2 \cup \mathbf{S}_3 \cup \mathbf{S}_4$$ generates $\Syz(\mathcal{F})$.
Notice moreover that all the elements of $\mathbf{S}$ are multigraded. In our notation, the multidegree of an element of $\mathbf{S}_1$ is
$$\deg(|D_{\{h\}\cup I,J}|_{1,\bullet})=e_h + \sum_{i \in I} e_i + \sum_{j \in J} f_j,$$
the multidegree of an elements of $\mathbf{S}_2$ is 
$$\deg(|D_{I,\{k\}\cup J}|_{\bullet,1})=\sum_{i \in I} e_i + f_k + \sum_{j \in J} f_j,$$
and that of an element of $\mathbf{S}_3\cup \mathbf{S}_4$ is $$\deg(|D_{I,J}|_{1,\bullet}-|D_{I,J}|_{h,\bullet})=\deg(|D_{I,J}|_{1,\bullet} - |D_{I,J}|_{\bullet,k})=
\sum_{i \in I} e_i + \sum_{j \in J} f_j.$$
Since the multidegrees of the elements of $\mathbf{S}_1\cup\mathbf{S}_2$ are pairwise distinct, then the elements of $\mathbf{S}_1\cup\mathbf{S}_2$ are $\mathbb{K}$-linearly independent. For the same reason, the elements of $\mathbf{S}_1\cup\mathbf{S}_2$ are linearly independent from those of $\mathbf{S}_3\cup\mathbf{S}_4$.

\medskip
In the next theorem we identify a system of generators of $\mathbf{U}_{r+2}$.
Notice that these are the syzygies identified in \cite{BBCGPSTV20}.

\begin{theorem}\label{thm:U}
Let
\begin{align*} 
\mathbf{S}'_1= &\left\{
|D_{\{h\}\cup I,J}|_{1,\bullet}
: \abs{I} =r+1, \abs{J} = r+2, I \cap [m]= \{h\} \right\},\\
\mathbf{S}'_3= &\left\{
|D_{I,J}|_{1,\bullet}-|D_{I,J}|_{2,\bullet}
: \abs{I}=\abs{J}=r+2, |I \cap [m]| = 2 \right\},
\end{align*}
where $I \subseteq [m+r]$ and $J \subseteq [n]$ are subsets. The set $\mathbf{S}':=\mathbf{S}'_1 \cup \mathbf{S}'_3$ generates $\mathbf{U}_{r+2}$ as a $\KK$-vector space.
\end{theorem}	 

\begin{proof}
It is easy to check that $\mathbf{S}^\prime_1\cup\mathbf{S}^\prime_3\subseteq\mathbf{U}_{r+2}$. We now prove that every element of $\mathbf{U}_{r+2}$ can be written as a linear combination of the elements of $\mathbf{S}^\prime_1\cup\mathbf{S}^\prime_3$.

Let $T\in\mathbf{U}_{r+2}$.
Since $\mathbf{U}$ is multigraded, by considering each homogeneous component we may assume without loss of generality that $T$ is multigraded.
Since $T\in\Syz(\mathcal{G})_{r+2}$, then $T$ is an element of $\Syz(\mathcal{F})_{r+2}$ which belongs to the submodule generated by $E_{H,K}$ with $H\supseteq\{m+1,\ldots,m+r\}$, i.e., an element whose multidegree is bigger than $\sum_{i=m+1}^{m+r} e_i+\sum_{j \in K} f_{j}$ for some subset $K\subseteq[n]$ of $|K|=r+1$. More precisely 
$$\deg(T)=e_h+\sum_{i=m+1}^{m+r} e_i + e_{i^*}+ \sum_{j \in J} f_{j}  + f_{j^{*}},$$
for some $h \in [m]$, $i^* \in [m+r]$, $J \subseteq [n]$, $|J|= r+1$, and $j^* \in [n]$.
Since $T\in\mathbf{U}$, then its entries only involve the variables $y_{i,j}$. This forces $i^* \in [m]$.

Since $T \in \mathbf{U}_{r+2}\subseteq \Syz(\mathcal{F})_{r+2}$, $T$ can be written as linear combination of elements in $\mathbf{S}$.
By comparing the degree of $T$ with those of the elements of $\mathbf{S}$, one sees that either $T\in\mathbf{S}_1$
or $T\in\langle\mathbf{S}_3\cup\mathbf{S}_4\rangle$.
In the first case, $i^*=h$ and $T\in\mathbf{S}^\prime_1$, since $\mathbf{S}^\prime_1$ consists of the elements of $\mathbf{S}_1$ with $I=\{h,m+1,\ldots,m+r\}$. In the second case,  $T\in\langle\mathbf{S}_3\cup\mathbf{S}_4\rangle$ has multidegree $\sum_{i\in I} e_{i}+\sum_{j\in J} e_{j}$ for some $I,J$ with $\{m+1,\ldots,m+r\}\subseteq I\subseteq [m+r]$ and $J\subseteq[n]$ of $|I|=|J|=r+2$. For ease of notation, denote by $S_{h,\bullet}=|D_{I,J}|_{1,\bullet}-|D_{I,J}|_{h,\bullet}\in\mathbf{S}_3$ and $S_{\bullet,k}=|D_{I,J}|_{1,\bullet} - |D_{I,J}|_{\bullet,k}\in\mathbf{S}_4$.
Write \begin{equation}\label{eqn:T}
T=\sum_{h=2}^{r+2}\alpha_h S_{h,\bullet}+\sum_{k=2}^{r+2}\beta_kS_{\bullet,k}
\end{equation} for some $\alpha_h,\beta_k\in\mathbb{K}$. 
The element $S_{\bullet,k}$ has an entry $c_{r,k}$ in position $I \setminus \{m+r\}, J\setminus \{k\}$ and no other element appearing in the sum~(\ref{eqn:T}) has a nonzero entry in the same position. For $3\leq h\leq r+2$, the element $S_{h,\bullet}$ has an entry $c_{h-2,1}$ in position $I\setminus \{m+h-2\}, J\setminus \{1\}$ and no other element in the sum~(\ref{eqn:T}) has a nonzero entry in the same position. Since $T$ does not involve the variables $c_{i,j}$, this proves that $\beta_k=0$ for $2\leq k\leq r+2$ and $\alpha_h=0$ for $3\leq h\leq r+2$. This yields the set $\mathbf{S}_3^\prime$. 
\end{proof}

\section{The general case}\label{sect:spec}

In this section we discuss the general case when the entries of $M_{\vec{x}}$ are linear forms in $x_1,\ldots,x_K$. 
Consider $P=\KK[x_\ell,c_{ij}: \ell\in[K], i\in[r], j\in[n]]$ and the $R$-algebra homomorphism $\rho:R\rightarrow P$ given by $c_{ij}\mapsto c_{ij}$ and $y_{kj}\mapsto m_{kj}(\vec{x})$, where $m_{kj}(\vec{x})$ are the entries of the matrix $M_{\vec{x}}$. Abusing notation, we denote by $\rho$ also the homomorphism $R^t \rightarrow P^t$, $t \in \mathbb{N}$, that acts as $\rho$ componentwise. We often use specializing as a synonym for computing the image of an object via $\rho$. 

We are interested in the syzygies of the $(r+1)$-minors of the matrix $$\rho(D)=\begin{pmatrix}
	M_{\vec{x}} \\
	C
\end{pmatrix}$$ 
which involve the last $r$ rows, that is, the syzygies of $\rho(\mathcal{G})$. In particular, we want to compute the module of syzygies of $\rho(\mathcal{G})$ that only involve the $x$-variables.
Since $\rho$ is a homomorphism, specializing the elements of $\mathbf{U}$ yields syzygies of $\rho(\mathcal{G})$ in the $x$-variables.
In other words,
$$\rho(\mathbf{U}) \subseteq \Syz(\rho(\mathcal{G}))\cap\mathbb{K}[x_1,\ldots,x_K] \subseteq P^{\binom{r+m}{r+1} \binom{n}{r+1}}.$$

\subsection{The case $b=2$}

In this subsection, we argue that there exists a subset of choices of coefficients of the linear forms $m_{jk}$ that is dense in the Zariski topology and for which $\rho(\mathbf{S}')$ generates $\Syz(\rho(\mathcal{G}))_{r+2}\cap\mathbb{K}[x_1,\ldots,x_K]$.
We start with a preliminary lemma on the supports of the syzygies before specialization. We follow the notation of the previous section. 
For brevity, we say that that $S\in\Syz(\mathcal{F})$ is supported on $\mathcal{H}$ to mean that $S$ is supported on the positions corresponding to $\mathcal{H}$, for $\mathcal{H}\subseteq\mathcal{F}$.

\begin{lemma}\label{lemma:supp}
Each $S\in \mathbf{S}$ falls into one of these mutually exclusive cases:
\begin{enumerate}[i)]
\item $S$ is supported on the positions corresponding to $\mathcal{G}$.
\item $S$ is supported on the positions corresponding to $\mathcal{F}\setminus \mathcal{G}$.
\item $S$ does not fall into case i) or ii) and it involves a variable $c_{ij}$ in a position that corresponds to an element of $\mathcal{F}\setminus \mathcal{G}$.
\end{enumerate}  
\end{lemma}

\begin{proof}
Let $S\in\mathbf{S}=\mathbf{S}_1\cup\mathbf{S}_2\cup\mathbf{S}_3\cup\mathbf{S}_4$.
If $S\in \mathbf{S}_1$, then $S=\abs{D_{\{h\}\cup I, J}}_{1,\bullet}$ where $I\subseteq[m+r]$, $J\subseteq[n]$ are subsets, $\abs{I}=r+1$, $\abs{J} = r+2$, and $h \in I$. 
In particular, $\abs{I\cap [m]}\ge 1$.
If $\abs{I\cap [m]}>1$, then $S$ is supported on $\mathcal{F}\setminus \mathcal{G}$ and it falls in case ii).
If $\abs{I\cap [m]}=1$, then $S$ is supported on $\mathcal{G}$ and it falls in case i).

If $S\in\mathbf{S}_2$, then $S=\abs{D_{I, \{k\}\cup J}}_{\bullet,1}$ where $I\subseteq[m+r]$, $J\subseteq[n]$ are subsets, $\abs{I} =r+2$, $\abs{J} = r+1$, and $k \in J$. Notice that $\abs{I\cap [m]}\ge 2$.
If $\abs{I\cap [m]}>2$, then $S$ is supported on $\mathcal{F}\setminus \mathcal{G}$ and it falls in case ii).
If $\abs{I\cap [m]}=2$, then $S$ is the sum with alternating signs over $i\in I$ of $d_{i k}E_{I\setminus \{i\},J}$, where $d_{i k}$ is the entry of $D$ in position $(i,k)$. Since $r\geq 1$, then $|I|\geq 3$, so there exists $i^*\in I\setminus[m]$. Let $\iota=\min(I)$. Then $E_{I\setminus \{\iota\},J}$ is supported on $\mathcal{G}$ and $d_{\iota k}=y_{\iota k}$. Moreover, $E_{I\setminus \{i^*\},J}$ is supported on $\mathcal{F}\setminus\mathcal{G}$ and $d_{i^* k} = c_{i^*-m k}$, so $S$ falls in case iii).

If $S\in \mathbf{S}_3$, then  $S=\abs{D_{I,J}}_{1,\bullet} - \abs{D_{I,J}}_{h,\bullet}$ where $I\subseteq[m+r]$, $J\subseteq[n]$ are subsets, $\abs{I}= \abs{J} =r+2 $, and $2\le h\le r+2$.
Notice that $\abs{I\cap [m]}\ge 2$.
If $\abs{I\cap [m]}>2$, then $S$ is supported on $\mathcal{F}\setminus \mathcal{G}$ and it falls in case ii).
If $\abs{I\cap [m]}=2$ and $h=2$, then $S$ is supported on $\mathcal{G}$ and it falls in case i).
If $\abs{I\cap [m]}=2$ and $h>2$, then 
\[S = \sum_{j\in J} (-1)^j d_{i_1 j}E_{I\setminus \{i_1\},J\setminus \{j\}} - \sum_{j\in J} (-1)^{j+h} d_{i_h j}E_{I\setminus \{i_h\},J\setminus \{j\}},\]
where $I=\{i_1,\ldots,i_{r+2}\}$ with $i_1<\cdots<i_{r+2}$.
Notice that $E_{I\setminus \{i_1\},J\setminus \{j\}}$ is supported on $\mathcal{G}$ and $d_{i_1 j} = y_{i_1 j}$ for all $j\in J$.
Moreover, $E_{I\setminus \{i_h\},J\setminus \{j\}}$ is supported on $\mathcal{F}\setminus\mathcal{G}$ and $d_{i_h j} = c_{i_h-m j}$ for all $j\in J$, so $S$ falls in case iii).

If $S\in \mathbf{S}_4$, then $S=\abs{D_{I,J}}_{1,\bullet} - \abs{D_{I,J}}_{\bullet,k}$ where $I\subseteq[m+r]$, $J\subseteq[n]$ are subsets, $\abs{I}= \abs{J} =r+2 $, and $2\leq k\leq r+2$.
Notice that $\abs{I\cap [m]}\ge 2$.
If $\abs{I\cap [m]}>2$, then $S$ is supported on $\mathcal{F}\setminus \mathcal{G}$ and it falls in case ii).
If $\abs{I\cap [m]}=2$, then
\[S = \sum_{t=1}^{r+2} (-1)^{t+1} d_{i_1 j_t}E_{I\setminus \{i_1\},J\setminus \{j_t\}} - \sum_{s=1}^{r+2} (-1)^{s+k} d_{i_s j_k}E_{I\setminus \{i_s\},J\setminus \{j_k\}},\]
where $I=\{i_1,\ldots,i_{r+2}\}$ with $i_1<\cdots<i_{r+2}$ and $J=\{j_1,\ldots,j_{r+2}\}$ with $j_1<j_2<\cdots<j_{r+2}$. 
For $s>2$, $E_{I\setminus \{i_s\},J\setminus \{j_k\}}$ is supported on $\mathcal{F}\setminus\mathcal{G}$ and $d_{i_s j_k} = c_{i_s-m, j_k}$. 
Hence $S$ falls in case iii).
\end{proof}

We now prove that every syzygy of $\mathcal{F}$ which specializes to a nonzero syzygy of $\rho(\mathcal{G})$ is in fact a syzygy of $\mathcal{G}$.

\begin{theorem}\label{thm:supp}
Let $P=\KK[x_\ell,c_{ij}: \ell\in[K], i\in[r], j\in[n]]$ be an $R$-algebra with homomorphism $\rho:R\rightarrow P$ given by $c_{ij}\mapsto c_{ij}$ and $y_{kj}\mapsto m_{kj}$.
Let $T\in \Syz(\mathcal{F})$. If $\rho(T)\in \Syz(\rho(\mathcal{G}))\setminus\{0\}$, then $T\in\Syz(\mathcal{G})$.
\end{theorem}

\begin{proof} 
Let $\mathbf{T}_1$ be the set of elements of $\mathbf{S}$ that fall into case i) of Lemma~\ref{lemma:supp}, let $\mathbf{T}_2$ be the set of those that fall into case ii), and let $\mathbf{T}_3$ be the set of those that fall into case iii). By Lemma~\ref{lemma:supp}, $\mathbf{S}=\mathbf{T}_1\cup\mathbf{T}_2\cup\mathbf{T}_3$. 
Let $T\in \Syz(\mathcal{F})$ and suppose that $\rho(T)\in \Syz(\rho(\mathcal{G}))$.
Since $\mathbf{S}$ generates $\Syz(\mathcal{F})$, by Lemma~\ref{lemma:supp} we can write
\begin{equation}\label{eqn:expT}
T=\sum_{S\in\mathbf{T}_1}\alpha_S S + 
\sum_{S\in\mathbf{T}_2}\alpha_S S + 
\sum_{S\in\mathbf{T}_3} \alpha_S S.
\end{equation}
Up to replacing $T$ by $T-\sum_{S\in\mathbf{T}_1}\alpha_S S$, we may assume that $\alpha_S=0$ for every $S\in\mathbf{T}_1$. In particular,
\begin{equation}\label{eqn:rhoT}
\rho(T) = \sum_{S\in\mathbf{T}_2}\alpha_S \rho(S) + 
\sum_{S\in\mathbf{T}_3} \alpha_S\rho(S).
\end{equation}
Our thesis corresponds to proving that $T=0$. 

We start by discussing how the support changes when passing from $T$ to $\rho(T)$. Since $\rho(c_{ij})=c_{ij}$ for all $i$ and $j$, if one entry of $T$ involves a $c$-variable with a nonzero coefficient, then the corresponding entry of $\rho(T)$ involves the same $c$-variable with the same coefficient. In particular, when looking at the $c$-variables, the supports of $T$ and $\rho(T)$ coincide. Since $\rho(T)\in\Syz(\mathcal{G})$, the positions of $T$ corresponding to elements of $\mathcal{F}\setminus\mathcal{G}$ cannot involve any $c$ variable.

A coefficient $c_{ij}$ in position $(I,J)$ can only come from an element of $\mathbf{T}_2$ or $\mathbf{T}_3$ corresponding to the multisets $I\cup\{i\}$ and $J\cup\{j\}$. Therefore, if $c_{ij}E_{I,J}$ comes from a syzygy $\rho(S)$ and cancels in (\ref{eqn:expT}), then it cancels with a summand $c_{ij}E_{I,J}$ coming from a different syzygy $\rho(S^\prime)$ which corresponds to the same multisets $I\cup\{i\}$ and $J\cup\{j\}$. Therefore, we may restrict to $D_{I\cup\{i\},J\cup\{j\}}$ and only discuss which cancellations occur in (\ref{eqn:expT}) for syzygies that originate from it. 

Let $S\in\mathbf{T}_3$ and let $I,J$ be the multisets from which $S$ originates. Then $I\subseteq[m+r]$, $J\subseteq[n]$, $|I|=|J|=r+2$ and one of the following must hold:
\begin{enumerate}
\item $I$ is a set with $|I\cap[m]|=2$, $J$ contains one element $k$ with multiplicity two and every other element has multiplicity one,  and $S\in\mathbf{S}_2$.
\item $I$ and $J$ are sets of cardinality $r+2$, $|I\cap[m]|=2$, and $S\in\mathbf{S}_3$ with $h>2$.
\item $I$ and $J$ are sets of cardinality $r+2$, $|I\cap[m]|=2$, and $S\in\mathbf{S}_4$ with $k>1$.
\end{enumerate}

In case 1., there exists exactly one syzygy $\Sigma\in\mathbf{T}_3$ that originates from those $I$ and $J$. Write $I=\{i_1,\ldots,i_{r+2}\}$ with $i_1<\ldots<i_{r+2}$, $J=\{k,j_1,\ldots,j_{r+1}\}$ with $j_1<\ldots<j_{r+1}$ and $k\in\{j_1,\ldots,j_{r+1}\}$. In this case we have $\alpha_{\Sigma}=0$, since otherwise the element $\alpha_{\Sigma}c_{i_3 k}E_{I\setminus\{i_3\},\{j_1,\ldots,j_{r+1}\}}$ does not cancel in (\ref{eqn:expT}), as there is exactly one syzygy in $\mathbf{T}_2\cup \mathbf{T}_3$ where $c_{i_3 k}$ appears in position $I\setminus\{i_3\},\{j_1,\ldots,j_{r+1}\}$.

In cases 2. and 3., assume for ease of notation that $I=\{1,2,m+1,\ldots,m+r\}$ and $J=[r+2]$.
Because of what we discussed above, when studying cancellations among the $c$-variables, we may restrict our attention to the syzygies that originate from the $I$ and $J$ that we just fixed. For $h>2$, $c_{h-2,1}E_{I\setminus\{m+h-2\},\{2,\ldots,r+2\}}$ only appears in syzygies obtained from developing the determinant of $D_{I,J}$ with respect to row $h$ or column $1$. The only syzygy in $\mathbf{T}_2\cup \mathbf{T}_3$ which involves $c_{h-2,1}$ is $\Sigma_h=|D_{I,J}|_{1,\bullet} - |D_{I,J}|_{h,\bullet}$. Therefore no cancellation is possible and $\alpha_{\Sigma_h}=0$ for $h>2$. For a fixed $2\leq k\leq r+2$, $c_{1,k}E_{I\setminus\{m+1\},J\setminus\{k\}}$ only appears in syzygies obtained from developing the determinant of $D_{I,J}$ with respect to row $3$ or column $k$.
Since $\alpha_{\Sigma_3}=0$, the only syzygy in which it appears is $\Theta_k=|D_{I,J}|_{1,\bullet} - |D_{I,J}|_{\bullet,k}$. Again, no cancelation is possible, showing that $\alpha_{\Theta_k}=0$ for $k>1$.

This proves that, if $\rho(T)\in\Syz(\mathcal{G})$, then there is an expression of $T$ as in (\ref{eqn:expT}) that does not involve any element of $\mathbf{T}_3$. Since the support of $\rho(S)$ is disjoint from $\mathcal{G}$ for every $S\in\mathbf{T}_2$, we deduce that $\rho(T)\in\Syz(\mathcal{G})$ forces $T=0$. 
\end{proof}

The next theorem is the main result of this paper. It shows that, for $K$ sufficiently large and for a generic choice of $M_1,\ldots,M_K$, the linear syzygies of $\rho(\mathcal{G})$ which only involve the $x$-variables are generated by the specializations of the linear syzygies of $\mathcal{G}$ which only involve the $y$-variables. We follow the same notation as in the rest of the section.

\begin{theorem}\label{thm:main}
For $K\geq m(n-r)$ and generic $M_1,\ldots,M_K$, the set $\rho(\mathbf{S}')$ generates $\rho(\mathbf{U})_{r+2}$ as a $\mathbb{K}$-vector space.
\end{theorem}

\begin{proof}
For $K\geq (m-1)(n-r-1)$ and generic $M_1,\ldots,M_K$, one has that $\grade(\rho(\mathcal{F}))=\grade((\mathcal{F}))$. Since the ideal $(\mathcal{F})\subseteq R$ is a perfect $R$-module, a minimal free $R$-resolution of $(\mathcal{F})$ specializes to a minimal free $P$-resolution of $(\rho(\mathcal{F}))$ by~\cite[Theorem 3.5]{BV88} (as tensoring with $P$ over $R$ corresponds to specializing via $\rho$).
In particular, $$ \Syz(\rho(\mathcal{F}))=\rho(\Syz(\mathcal{F})).$$ 

Since $\rho$ is a homomorphism, then $\rho(\Syz(\mathcal{G}))\subseteq\Syz(\rho(\mathcal{G}))$. Conversely, let  $0\neq S\in\Syz(\rho(\mathcal{G}))\subseteq\Syz(\rho(\mathcal{F}))=\rho(\Syz(\mathcal{F}))$. Then there is $T\in\Syz(\mathcal{F})$ such that $\rho(T)=S\in\Syz(\rho(\mathcal{G}))$. By Theorem \ref{thm:supp}, $T\in\Syz(\mathcal{G})$. This proves that 
$$\rho(\Syz(\mathcal{G}))=\Syz(\rho(\mathcal{G})).$$

Finally, let $0\neq S\in\Syz(\rho(\mathcal{G}))_{r+2}\cap\mathbb{K}[x_1,\ldots,x_K]^{m\binom{n}{r+1}}$ and let $T\in\Syz(\mathcal{G})$ be such that $\rho(T)=S$. Since $\rho$ is the identity on the $c$-variables and maps the $y$-variables into linear forms in the $x$-variables and $S\in\mathbb{K}[x_1,\ldots,x_K]^{m\binom{n}{r+1}}$, then $T\in\mathbb{K}[y_{kl}: k\in[m], l\in[n]]^{m\binom{n}{r+1}}$. Therefore $T\in\Syz(\mathcal{G})\cap\mathbb{K}[y_{kl}: k\in[m], l\in[n]]^{m\binom{n}{r+1}}=\mathbf{U}$. This proves that $$\Syz(\rho(\mathcal{G}))_{r+2} \cap\mathbb{K}[x_1,\ldots,x_K]^{m\binom{n}{r+1}}=\rho(\mathbf{U})_{r+2}= \langle\rho(\mathbf{S}^\prime) \rangle,$$
where the last equality follows from Theorem~\ref{thm:U}.
\end{proof}

\subsection{Submaximal minors}

In this subsection, we discuss the special case of maximal minors, that is, the case when $r=n-1$. For the convenience of the reader, we start by recalling a results from commutative algebra, which we use in the sequel. Then we apply it to our situation and we use it to compute the dimension of the module $\rho(U)$ in every degree.

In~\cite{AS81}, Andrade and Simis give a free resolution of the ideal generated by the ideal of maximal minors of an $n\times\ell$ matrix, which involve a fixed set of $n-1$ columns.
Their free resolution is a modification of the Buchsbaum-Rim complex, a well-studied complex in commutative algebra.   
  
\begin{theorem}[{\cite[Theorem]{AS81}}]\label{thm:free:res}
Let $R$ be a Noetherian ring and let $f:F\rightarrow G$ be an $R$-homomorphism  of free modules, with $\mathrm{rank}(F)=\ell \geq \mathrm{rank}(G)=n$. Let $F=F' \oplus F''$ be a decomposition of $F$ as direct sum of free modules, with $\mathrm{rank}(F')=n-1$. Set $f':F'\rightarrow G$ for the restriction map. The following are equivalent:
\begin{enumerate}
\item[$(i)$] $\mathrm{grade}(I_n)(f)\geq \ell-n+1$ and $\mathrm{grade}(I_{n-1}(f'))\geq 2$,
\item[$(ii)$] the complex 
\begin{multline} \label{compl:G}
0 \rightarrow \bigwedge^\ell F \otimes S_{\ell-n-1}(G^{*}) \xrightarrow{d_{\ell-n+1}} \cdots \\ \xrightarrow{d_3} \bigwedge^{n+1} F \otimes S_{0}(G^{*}) \rightarrow F'' \xrightarrow{\varphi} \mathrm{Im}(f)/\mathrm{Im}(f') \rightarrow 0 
\end{multline}
is exact and $\mathrm{Im}(f)/\mathrm{Im}(f') \cong (\bigwedge^n f)(F'')$, where $F''$ sits naturally inside $\bigwedge^n F = \bigwedge^n (F' \oplus F'')=\bigoplus_{i=0}^n (\bigwedge^i F' \otimes \bigwedge^{n-1} F'')$ as $\bigwedge^{n-1} F' \otimes \bigwedge^1 F''$.
\end{enumerate}
\end{theorem}

Thanks to the identification of $F''$ with $\bigwedge^{n-1} F' \otimes \bigwedge^1 F''$,  $(\bigwedge^n f)(F'') \subseteq R$ is the ideal of $R$ generated by the maximal minors which involve the $n-1$ columns corresponding to $F'$ in the matrix representing $f$. In fact, the image via $\varphi$ of the element $e_i$ of the standard basis of $F''$ is the maximal minor of the matrix representing $f:F \rightarrow G$ which involves the $n-1$ columns corresponding to $F'$ and the $i$-th column among those that correspond to $F''$. In the following example, we illustrate how the identification works.

\begin{example}
Let $R=\KK[x,y,z]$, $F=R^4$, and $G=R^3$ with bases $\{e_1,e_2,e_3,e_4\}$ and $\{f_1,f_2,f_3\}$, respectively. Let $f: R^4 \rightarrow R^3$ be represented by the matrix 
\begin{equation*}
\begin{pmatrix}
x & 0 & 0 & yz \\
0 & y & 0 & x \\
0 & 0 & z & y^2
\end{pmatrix}.
\end{equation*}
Setting $F'=\langle e_3, e_4 \rangle$ and $F''=\langle e_1, e_2 \rangle$, the columns corresponding to $F'$ are the last two. Thanks to the isomorphisms $$F'' \cong \bigwedge^2 F' \otimes \bigwedge^1 F'' \cong \langle e_1 \wedge e_3 \wedge e_4 , e_2 \wedge e_3 \wedge e_4 \rangle \subseteq \bigwedge^3 R^4,$$ one has that  $\bigwedge^3f (F'') = (x^2 z, y^2z^2)\subseteq R$, that is, the ideal generated by the $3$-minors of $M$ which involve the last two columns.
\end{example} 

We now apply Theorem~\ref{thm:free:res} to our situation. Let $r=n-1$ and let $P=\KK[x_{\ell},c_{ij}: \ell\in[K], i\in[n-1], j\in[n]]$. Let $f:P^{m+r} \rightarrow P^n$ be the $P$-module homomorphism represented by the $n \times (m+r)$ matrix
$\begin{pmatrix}
M_x \lvert C
\end{pmatrix}$.
The polynomial ring $P$ can be given a standard $\ZZ^2$-grading by setting $\deg(x_{\ell})=(1,0) \in \ZZ^2$ and $\deg(c_{ij})=(0,1) \in \ZZ^2$. 
Let $F=P(-(1,0))^{m} \oplus P(-(0,1))^{n-1}$ be a graded free $P$-module with decomposition $F' \oplus F''$, where $F'=P(-(0,1))^{n-1}$ and $F''=P(-(1,0))^{m}$, and let $G=P^n$. Then $$\left(\bigwedge^n f\right)(F'') = (\rho(\mathcal{G})),$$
where $\rho(\mathcal{G})$ is the set of maximal minors of $\begin{pmatrix}
M_x \lvert C
\end{pmatrix}$ which involve all the columns of $C$.

\begin{corollary}\label{lemma:free:res:Syz:G}
Let $\rho(\mathcal{G})$ be the set of maximal minors of $\begin{pmatrix}
M_x \lvert C
\end{pmatrix}$ which involve all the columns of $C$.
For $K\geq m$ and a generic choice of $M_1,\ldots,M_K$, the complex 
\begin{equation} \label{compl:graded:syzG}
0 \rightarrow \bigwedge^{m+n-1} F \otimes S_{m-2}(G^{*}) \xrightarrow{d_{m-2}} \cdots \xrightarrow{d_3} \bigwedge^{n+1} F \otimes S_{0}(G^{*}) 
\end{equation}
is a graded free resolution of the $P$-module $\Syz(\rho(\mathcal{G}))(0,n-1)$.
\end{corollary}

\begin{proof}
Notice that the choice of grading makes the function $f$ and all the maps $d_i$ in~\eqref{compl:G} homogeneous. In order to make the function $\varphi: P(-(1,0))^{m} \rightarrow (\rho(\mathcal{G}))$ homogeneous, one needs to consider the shifted ideal $(\rho(\mathcal{G}))(0,n-1)$ instead of $(\rho(\mathcal{G}))$. 

By~\cite[Theorem 2.5]{BV88} and by the genericity of the choice of $M_1,\ldots,M_K$, $\mathrm{grade}(I_n(f))=m$ and $\mathrm{grade}(I_{n-1}(f'))=2$ . Then by Theorem~\ref{thm:free:res} the complex 
\begin{equation} \label{compl:graded:G}
0 \rightarrow \bigwedge^{m+n-1} F \otimes S_{m-2}(G^{*}) \xrightarrow{d_{m}} \cdots \xrightarrow{d_3} \bigwedge^{n+1} F \otimes S_{0}(G^{*}) \rightarrow F'' 
\end{equation}
is a graded free resolution of the shifted ideal $(\rho(\mathcal{G}))(0,n-1)$ with augmentation map $\varphi$.  
The thesis now follows from the exactness of~\eqref{compl:graded:G}, since $\ker(\varphi)=\Syz(\rho(\mathcal{G}))(0,n-1)$.
\end{proof}

Finally, in the next theorem we compute the dimension as $\KK$-vector space of the space of syzygies $\rho(\mathbf{U})_{n-1+b}$.

\begin{theorem} \label{thm:dim:submax} 
For $K\geq m+1 - (n-1)n$, generic $M_1,\ldots,M_K$, and any $b>0$
$$\dim_{\KK}(\rho(\mathbf{U})_{n-1+b})=\sum_{i=1}^{\min\{m-n,n+1,b-n\}} (-1)^{i-1} \binom{m}{n+i}  \binom{n}{i-1} \binom{K + b-n -i -1}{K-1}.$$
\end{theorem}

\begin{proof}
First notice that $$\rho(\mathbf{U})_{n-1+b}=\Syz(\rho(\mathcal{G}))_{(b,n-1)}=\Syz(\rho(\mathcal{G}))(0,n-1)_{(b,0)}.$$ In order to compute its dimension, we use the homogeneous component of degree $(b,0)$ of the graded free resolution~\eqref{compl:graded:syzG}
\begin{multline}\label{eqn:ressyz}
0 \rightarrow \left(\bigwedge^{m+n-1} F \otimes S_{m-2}(G^{*})\right)_{(b,0)} \xrightarrow{d_{m-2}} \cdots \\ \xrightarrow{d_3} \left(\bigwedge^{n+1} F \otimes S_{0}(G^{*}) \right)_{(b,0)}\rightarrow\Syz(\rho(\mathcal{G}))_{(b,n-1)}\rightarrow 0.
\end{multline}
It can be shown by direct computation that the module in the $i$-th homological position is $$\bigwedge^{n+i} F \otimes S_{i-1}(G^{*})\cong \bigoplus_{j=0}^{n+i}P(-j,j-n-i)^{\beta_{i,(-j,j-n-i)}}$$
where $\beta_{i,(-j,j-n-i)}={m\choose j}{n-1\choose n+i-j}{n\choose i-1}.$ Therefore
\begin{align*}
\dim_{\KK}\left(\bigwedge^{n+i} F \otimes S_{i-1}(G^{*})\right)_{(b,0)} &=\sum_{j=0}^{n+i} \beta_{i,(-j,j-n-i)} \dim_{\FF_q} P_{(b-j,j-n-i)} \\
&= \beta_{i,(-n-1,0)} \binom{K + b -n -i -1}{K-1} \\
&= \binom{m}{n+i}  \binom{n}{i-1} \binom{K + b-n -i -1}{K-1},
\end{align*}
where the second equality follows from observing that $\dim_{\KK} P_{(b-j,j-n-i)} \neq 0$ if and only if $j \geq n+i$.  

By the additivity of dimension on (\ref{eqn:ressyz}), one obtains
\begin{align*}
\dim(\rho(\mathbf{U})_{n-1+b}) & =
\dim_{\KK}\Syz(\rho(\mathcal{G}))_{(b,n-1)} \\ & =\sum_{i=1}^{\min\{m-n,n+1,b-n\}} (-1)^{i-1} \binom{m}{n+i}  \binom{n}{i-1} \binom{K + b-n -i -1}{K-1},
\end{align*}
since all three binomials are different from zero if and only if $i \leq \min\{m-n, n+1, b-n\}$. 
\end{proof}

\section{Complexity estimates and conclusions}\label{sect:estimates}

In this section, we apply the results from the previous sections to the study of the complexity of the SupportMinors Algorithm described in~\cite[Sections 5.2 and 5.3]{BBCGPSTV20}. The results contained in the previous sections allow us to make the estimates of~\cite[Sections 5.2 and 5.3]{BBCGPSTV20} rigorous in some cases of cryptographic interest.

Let $C_{[r],J}$ be the Pl\"ucker coordinates of the matrix $C$, i.e., for any $J\subseteq [n]$ of $|J|=r$, $C_{[r],J}$ corresponds to the maximal minor of $C$ of columns indexed by $J$. Let $T=\KK[x_1,\ldots,x_K,C_{[r],J}\mid J\subseteq[n], |J|=r]$. Notice that the Pl\"ucker coordinates are not distinct variables, but they satisfy homogeneous quadratic equations called the Pl\"ucker relations.
The SupportMinors Algorithm consists of solving system (\ref{eqn:sm}) in $x_1,\ldots,x_K$ and the Pl\"ucker coordinates $C_{[r],J}$. The system is solved via a linearization technique, by performing Gaussian eliminations in a Macaulay matrix whose rows correspond to the multiples of the original equations by the monomials in $x_1,\ldots,x_K$ of degree $b-1$, for a suitable $b$. In addition, one of the $C_{[r],J}$'s is set to one. This works in practice, as it is equivalent to assuming that the corresponding minor of $C$ is invertible (which is true with high probability over a sufficiently large field). Notice that, while the $C_{[r],J}$ are not distinct variables, they can be treated as such for the purpose of this algorithm, since the degree in $C_{[r],J}$ of the equations that we consider is at most one.
In order to bound the complexity of the algorithm, one needs to estimate the rank of the corresponding  Macaulay matrix, or equivalently, the dimension of the module generated by the elements of $\mathcal{\rho(G)}$ over $\KK[x_1,\ldots,x_K]$ for a given degree $b$ in $x_1,\ldots,x_K$ and degree one in the Pl\"ucker coordinates. In the sequel, we also say that $b$ is the degree in $x$ of the equations that we consider. 

\begin{notation}
Throughout this section, when we say that the entries of $M_{\bf x}$ are generic, we mean that the coefficients of the entries of $M_{\bf x}$ belong to the Zarisky-dense open set considered in Section~\ref{sect:spec}.
\end{notation}

The following theorem proves that the heuristic estimates of~\cite[Sections 5.2 and 5.3]{BBCGPSTV20} are correct in the case $b=1,2$.

\begin{theorem}\label{thm:cpxty}
Consider the SupportMinors Algorithm and assume that the entries of $M_{\bf x}$ are generic. Let $b$ denote the degree in $x$ of the equations that we consider. 
Then the number of linearly independent equations available for linearization for $b=1,2$ is as predicted in~\cite[Sections 5.2 and 5.3]{BBCGPSTV20}, namely it is 
$$\min\left\{ m{n\choose r+1}, K{n\choose r}\right\}$$ for $b=1$. For $b=2$ assume that $m{n\choose r+1}\leq K{n\choose r}$. Then
the number of linearly independent equations available for linearization is $$\min\left\{ Km{n\choose r+1}-{m+1\choose 2}{n\choose r+2}, {K+1\choose 2}{n\choose r}\right\}.$$ 
\end{theorem}

\begin{proof}
The first formula in the statement follows from observing that the cardinality of $\rho(\mathcal{G})$ is $m{n\choose r+1}$ and the number of Pl\"ucker coordinates of $C$ is ${n\choose r}$, hence $\dim(T_{(1,1)})=K{n\choose r}$. 

To prove the second formula, we need to estimate the dimension in bidegree $(2,1)$ of the vector space generated by $\rho(\mathcal{G})\KK[x_1,\ldots,x_K]_1$. Notice that, since the Pl\"ucker relations are homogeneous quadratic relations, we may treat the Pl\"ucker coordinates as independent variables. For a generic $M_{\bf x}$, the elements of $\rho(\mathcal{G})$ are linearly independent provided that $m{n\choose r+1}\leq K{n\choose r}$. In such a situation, $$\dim(\langle\rho(\mathcal{G})\rangle)=|\rho(\mathcal{G})|=m{n\choose r+1}.$$ 
Since $\dim(\KK[x_1,\ldots,x_K]_1)=K$, then the dimension of the vector space generated by $\rho(\mathcal{G})\KK[x_1,\ldots,x_K]_{1}\subseteq T_{(2,1)}$ is the minimum between the dimension of $T_{(2,1)}$ and $$Km{n\choose r+1}-\dim(\Syz(\rho(\mathcal{G}))_{r+2}\cap\mathbb{K}[x_1,\ldots,x_K]).$$
The thesis now follows since $\dim(T_{(2,1)})={K+1\choose 2}{n\choose r}$
and 
$$\dim(\Syz(\rho(\mathcal{G}))\cap\mathbb{K}[x_1,\ldots,x_K]_{r+2})={m+1\choose 2}{n\choose r+2}$$ by Theorem~\ref{thm:main}.
\end{proof}

\begin{corollary}
Assume that the entries of $M_{\bf x}$ are generic and let $b$ denote the degree in $x$ of the equations that we consider. Then the SupportMinors Algorithm outputs a solution to MinRank in degree $b=1$ provided that $$m{n\choose r+1}\geq K{n\choose r}-1.$$
If the SupportMinors Algorithm does not output a solution to MinRank in degree $b=1$, then it outputs one in degree $b=2$ provided that 
$$Km{n\choose r+1}-{m+1\choose 2}{n\choose r+2}\geq {K+1\choose 2}{n\choose r}-1.$$
\end{corollary}

Notice that the case $b=2$ is of high interest, as this is the relevant degree in the attacks to ROLLO-I-256 and many instances of GeMSS, see~\cite[Sections 6.1 and 6.2]{BBCGPSTV20}.

\bibliographystyle{abbrv}
\bibliography{Bib_MinRank}
	
\end{document}